\documentclass{article}

\usepackage{arxiv}

\usepackage[utf8]{inputenc} 
\usepackage[T1]{fontenc}    
\usepackage{hyperref}       
\usepackage{url}            
\usepackage{booktabs}       
\usepackage{amsfonts}       
\usepackage{nicefrac}       
\usepackage{microtype}      
\usepackage{lipsum}
\usepackage{graphicx}
\graphicspath{ {./images/} }

\usepackage{cite}
\usepackage{amsmath,amssymb,amsfonts, amsthm}
\usepackage{algorithmic}
\usepackage{graphicx}
\usepackage{textcomp}
\usepackage{xcolor}
\def\BibTeX{{\rm B\kern-.05em{\sc i\kern-.025em b}\kern-.08em
    T\kern-.1667em\lower.7ex\hbox{E}\kern-.125emX}}
\usepackage{amsmath}
\usepackage{amssymb}
\usepackage[linesnumbered,ruled,vlined]{algorithm2e}
\usepackage{xcolor}
\usepackage{cleveref}
\usepackage{soul}
\usepackage{booktabs} 
\usepackage{cite}
\usepackage{adjustbox}
\usepackage{enumitem}
\usepackage{tikz}
\usepackage{multirow}
\usetikzlibrary{calc,arrows}
\usetikzlibrary{positioning}
\tikzset{block/.style={draw,thick,text width=5em,minimum height=1mm,align=center},
         line/.style={-latex}
}
\definecolor{mycolor}{rgb}{0.8,0.6,0.7}
\definecolor{mycolor1}{rgb}{0.8, 0.6, 0.5}
\definecolor{mycolor2}{rgb}{0.529,0.807,0.89}
\definecolor{mycolor3}{rgb}{0.3, 0.5, 0.6}

\usepackage{blindtext,graphicx}
\usepackage[absolute]{textpos}
\newcommand{\proto}{FASTEN}

\newtheorem{claim}{Claim}

\newcommand{\sd}[1] { \textcolor{blue}{{\bf SD: }{``\em #1''}}}
\crefname{property}{property}{properties}
\Crefname{property}{Property}{Properties}
\newtheorem{definition}{Definition}

\title{FASTEN: Fair and Secure Distributed Voting Using Smart Contracts}

\author{
Sankarshan Damle and Sujit Gujar \\
  Machine Learning Lab\\
  \textit{International Institute of Information Technology} (IIITH)\\
  Hyderabad, India \\
  \texttt{sankarshan.damle@research.iiit.ac.in; sujit.gujar@iiit.ac.in} \\
   \And
 Moin Hussain Moti \\
The Hong Kong University of Science and Technology\\
  Hong Kong \\
  \texttt{mhmoti@cse.ust.hk} \\
}

\begin{document}
\maketitle

\begin{textblock}{15}(3.25,1)
\noindent\scriptsize A version of this paper will appear in the \textit{IEEE International Conference on Blockchain and Cryptocurrency} (IEEE ICBC 2021). This is a full version.
\end{textblock}

\begin{abstract}
Electing democratic representatives via voting has been a common mechanism since the 17\textsuperscript{th} century. However, these mechanisms raise concerns about fairness, privacy, vote concealment, fair calculations of tally, and proxies voting on their behalf for the voters. Ballot voting, and in recent times, electronic voting via electronic voting machines (EVMs) improves fairness by relying on centralized trust. Homomorphic encryption-based voting protocols also assure fairness but cannot scale to large scale elections such as presidential elections. In this paper, we leverage the blockchain technology of distributing trust to propose a smart contract-based protocol, namely, \proto.  There are many existing protocols for voting using smart contracts. We observe that these either are not scalable or leak the vote tally during the voting stage, i.e., do not provide vote concealment.
In contrast, we show that \proto\ preserves voter's privacy ensures vote concealment, immutability, and avoids double voting. We prove that the probability of privacy breaches is negligibly small. Further, our cost analysis of executing \proto\ over Ethereum is comparable to most of the existing cost of elections. 
\end{abstract}


\section{Introduction}

Elections are fundamental to democratic governance. Since direct democracy - a form of government in which political decisions are made directly by the entire body of qualified citizens - is impractical, societies select representatives for governance. Consequently, elections have been a common mechanism for modern representative democracy since 17\textsuperscript{th} century \cite{election-britannica}. They make it possible to include every eligible individual in the decision-making process by registering its vote into the system. A \emph{fair election} is possible only when the voter can freely vote for its desired preference.

One must also ensure an agent's participation in the voting process is hidden. This can be achieved by eliminating the \textit{link} between the voter and its vote, i.e., \emph{anonymous} voting. In order to design such a fair election with anonymous voting, i.e., fair and secure election (FSE), we first define the following essential properties. 
\begin{enumerate}[leftmargin=*]
\item \textbf{Voter Anonymity (VA).} A vote cannot be traced back to the voter either during or after the election.
\item \textbf{Vote Concealment (VC).} The vote's value should remain hidden from the system (voters, candidate, election commission). This in turn ensures that the vote tally remains a mystery to the system until the voting window has expired.
\item \textbf{Vote Immutable (VI).} Once a voter casts its vote, it should be impossible to alter it to any other vote by anyone.
\item \textbf{Double Voting Inhibition (DVI).} A voter should be allowed to vote only once in a specific election. 
\end{enumerate}

\smallskip
\noindent\textit{FSE Overview.} Towards FSE, the most traditional voting method is \emph{paper ballots}. It partially ensures anonymity, vote concealment, and vote immutability. The major drawback of ballot-based voting is that it involves tiresome manual work in counting the votes. Along with the risk of unintentional and intentional human-error involved, the non-durability of paper and lack of a robust mechanism to avoid double voting are some of the other challenges involved with this system. 

Election through \emph{electronic voting machines} (EVM)s is a technological upgrade over the paper ballot system. EVMs provide voter anonymity that does not take voter ID as a parameter. But, it fails at guaranteeing vote immutability. This is because the voter needs to trust the company that provides the EVM software for vote concealment and vote immutability. They also entrust the company with shipping the EVMs with the correct version of the firmware, and thus, the EVM remains a black box to the voter. Besides, the double voting inhibition problem is still there.

Micali \emph{et al.} \cite{micali14} propose a protocol for secure auctions that is also applicable for voting. However, the body conducting the elections, \textit{election commission} (EC), will know all the votes post the voting stage. To overcome the ``centralized trust" placed on EC, the authors propose an expensive zero-knowledge proof of the result. Consequently, the protocol is also not scalable. Adida \cite{adida08} proposes a most popular scheme \emph{Helios}, which relies on the security of one server and is not viable for nationwide elections.  
Thus, there is a need to look for a completely different approach to conduct an FSE.

\smallskip
\noindent\textit{FSE over Blockchain.} With Ethereum~\cite{buterin14ethereum}, we observe that blockchain can not only be used to solve the problem of designing a cryptocurrency but can also be used to implement \emph{smart contracts}~\cite{buterin14ethereum}. A smart contract allows blockchain to establish an interactive platform for $n$ parties. Such a contract enforces the outcome of any event through a set of rules. These rules correspond to a programming language understandable to the execution system. The key concept here is distributing trust rather than relying on a single party. Thus, we explore ways to leverage such a distributed trust to conduct an FSE.

The important steps in the voting procedure for FSE are: (i) \textit{voter authentication}, i.e., a person claiming to be a voter should be an eligible voter; (ii) \textit{vote registration}, which preserves the privacy of the voter as well as its vote; (iii) \textit{outcome verification} that counts the tally of votes in a verifiable manner. Zhao \emph{et al.} \cite{zhao15vote} propose a voting protocol based on \emph{Commit-Publish} mechanisms that also leverages smart contract. As the authors' main goal is boardroom voting, it does not address step (i). It solves steps (ii) and (iii). To the best of our knowledge, in a plethora of voting schemes over blockchain, except \cite{yu18}, no protocol satisfies all the four desirable properties of a fair election. The challenge with \cite{yu18} is scalability as it is Zcash-based, which is \textit{considerably} slower than Bitcoin.

\smallskip

\noindent\textit{Our Approach.} We propose a novel protocol for FSE, namely, \proto: \textit{FAir and Secure disTributEd votiNg}. We partially rely on EC for authentication, which issues a random but \textit{unique} token for each voter after authenticating the voter. The token is unique to the voter and the particular election. If the voter tries to obtain multiple tokens for the same election - it will receive the same token. Therefore, \proto\ is resistant to \emph{Sybil} attacks. This authentication is similar to several secure applications over blockchain~(e.g., \cite{lu2018zebralancer,duan2019aggregating,damleaamas,damlearxiv}). In this work, we assume that EC does not store the link between a voter and its token\footnote{Such trusted third party authentications also have a close parallel with the Zcash Parameter Generation~\cite{zcashceremony}. These links thus correspond to ``toxic waste" -- to be destroyed.}. Next, the smart contract considers tokens as an eligibility to vote granted by EC. After this, each voter registers its encrypted vote and token to the smart contract. The smart contract\footnote{In this work we present  \proto\ as an Ethereum-based smart contract. The protocol can also be easily conducted as through computational logic over other distributed platforms like \textit{Hyperledger Fabric}~\cite{androulaki2018hyperledger} or \textit{Quorum}~\cite{baliga2018performance}.  } holds the hashes of all the tokens that are issued by EC. It computes the hash of the token registered by the voter and checks if it is present in the database. Once the entry is confirmed, the encrypted votes are registered in the voter database. After the vote casting window expires, the smart contract decrypts all the votes and computes the tally.

Since our protocol deploys smart contracts based on blockchain, one may also implement it as a \emph{Decentralized Application} (DApp). DApps comprise a friendly UI for any smart contract, thereby allowing layperson to interact with said contracts. Thus, \proto\ helps improve fairness in elections, i.e., designing FSE and improving voter participation. 

With \proto, we show that a straightforward protocol will achieve an FSE. We believe that the simple design of our protocol is desirable and will benefit the end-user. This simplicity is in contrast with sophisticated protocols which deploy heavy cryptographic and security primitives to achieve FSE. These protocols, especially how they reach the desired privacy, are challenging to explain to a layperson. These limit their use in general elections where participating voters are in millions. We outline some of them next.


\section{Related Work}
There have been few attempts before to design a voting protocol based on blockchains but they all use \emph{Commit-Publish} mechanisms to conceal the vote tally. A \emph{Commit-Publish} mechanism requires voters to submit a deposit before committing to a vote and are refunded the deposit when they publish their vote. Both the Open Vote Network Protocol implemented by McCorry \cite{mccorry17scvote} and the bitcoin based protocol by Zhao \cite{zhao15vote} use this mechanism. While this mechanism could be used for boardroom voting, it cannot be used for mass election like general presidential election of nations because most voters would be reluctant to submit the deposit. This is because, in a population of millions voters will not consider their vote to be of much importance to begin with. In addition, once the voters commit, they have to participate in the protocol again to reveal their commitment. Note that, several nations are struggling to bring voters to voting booths to cast their vote when they just have to come once. Thus, we believe that it is unrealistic to believe that a voter will choose to participate twice in the same election. With \proto, we require the voters to come on-chain, i.e., online on blockchain, only once to register their votes and without any deposit.

In other works, FollowMyVote \cite{followmyvote} uses trusted authorities to ensure voter anonymity. The voters cast their votes in plain-text therefore risking the concealment of the vote tally. In contrast, voters in \proto\ are required to submit encrypted votes, ``cipher-texts", thereby concealing their votes. Our encryption-decryption mechanism ensures the concealment of the vote tally until the vote casting window expires. Due to space constraints, we present several other works through Table~\ref{tab:comparison} to place \proto\ with respect to the existing literature in FSEs. We also refer the interested reader to \cite[Table 1]{survey2020} for a more comprehensive comparison of secure voting protocols over blockchain.
We next present the preliminaries required for the design and analysis of \proto.

\begin{table}[!t]
\begin{center}
\renewcommand*{\arraystretch}{1.25}
\begin{tabular}{cccccc}
\hline
\textbf{Paper} & \textbf{VA} & \textbf{VC} & \textbf{VI} & \textbf{DVI} &\textbf{Scalable} \\ \hline\hline
\cite{micali14} & \multirow{2}{*}{\checkmark} &  \multirow{2}{*}{$\times$}  & \multirow{2}{*}{\checkmark} & \multirow{2}{*}{\checkmark}&  \multirow{2}{*}{$\times$}\\
(No Blockchain) &  &  &  &  &  \\ \hline
\cite{adida08} & \multirow{2}{*}{\checkmark} & \multirow{2}{*}{\checkmark} &\multirow{2}{*}{\checkmark} & \multirow{2}{*}{\checkmark} &  \multirow{2}{*}{$\times$}\\ 
(No Blockchain) &  &  &  &  &  \\ \hline
\cite{mccorry17scvote,tarasov17,zhao15vote} & \checkmark & \checkmark & \checkmark & \checkmark & $\times$ \\ \hline
\cite{bellini19,followmyvote,hjalmarsson18,yi19} & \checkmark & $\times$ & \checkmark & \checkmark & \checkmark \\ \hline
\cite{yu18} & \checkmark & \checkmark & \checkmark &  \checkmark & $\times$ \\ \hline
\textit{\proto} & \checkmark & \checkmark & \checkmark & \checkmark & \checkmark\\ 
\end{tabular}
\label{tab:comparison}
\caption{Comparison of Different secure voting protocols}
\end{center}
\end{table}

\section{Model and Preliminaries}
The main stakeholders of any election are the voters, the candidates and the election commission (EC). EC represents the governing body of the election. Once an election is announced, interested candidates need to register their candidature with EC before a certain deadline which we denote as $t_{ecr}$. It is EC's responsibility to ensure fair candidature registration. For the time period between the beginning and end of token distribution, i.e., $t_{btd}$ and $t_{etd}$, EC issues tokens to the voters after authenticating them. Actual vote casting begins at time $t_{bvc}$ and ends at $t_{evc}$. In \proto, the voters submit their encrypted votes over the Ethereum-based smart contract. We show how EC (or any other interested party) can get the vote tally using the smart contract. We refer to the procedural methods of the overall protocol which are on the smart contract as \emph{on-chain methods} and the remaining procedures as \emph{off-chain methods}. Figure \ref{fig:voting} illustrates the temporal aspect of \proto. 

\begin{figure*}[!t]
\centering \scriptsize
\begin{tikzpicture}

   \node[block, line width=0.1mm, minimum height=4em, minimum width=7.5em, fill=red!10!white] (e1) at (-26,0) {\textbf{Apply For Candidature}};      
   \node[block, line width=0.1mm, minimum height=4em, minimum width=7.5em, fill=red!10!white] (e2) at (-24,0) {\textbf{Backout}};  
   \node[block, line width=0.1mm, minimum height=4em, minimum width=7.5em, fill=red!10!white] (e3) at (-22,0) {\textbf{Get Token}};
   
    \node[draw=none] (P2) at (-20,-4.25) {}; 
   
   \draw[line, line width=0.05mm, dashed, -] (P2.south) -- ([xshift=2cm,yshift=2cm]e3.south); 
 
 
    \node[block, line width=0.1mm, minimum height=4em, minimum width=7.5em, fill=blue!10!white] (f1) at (-18,0) {\textbf{Get Encryption Key}};
    \node[block, line width=0.1mm, minimum height=4em, minimum width=7.5em, fill=blue!10!white] (f2) at (-16,0) {\textbf{Get Candidate List}};
    \node[block, line width=0.1mm, minimum height=4em, minimum width=7.5em, fill=blue!10!white] (f3) at (-14,0) {\textbf{Cast Vote}};
    \node[block, line width=0.1mm, minimum height=4em, minimum width=7.5em, fill=blue!10!white] (f4) at (-12,0) {\textbf{Get Decryption Key}};
    \node[block, line width=0.1mm, minimum height=4em, minimum width=7.5em, fill=blue!10!white] (f5) at (-10,0) {\textbf{Tally Vote}};
        
  

    \node[circle, line width=0.1mm, draw=black, fill=red, text=black] (Can1) at (-26,-3.5) {C};
  
  \draw[line, line width=0.05mm] (Can1.north)-- (e1.south) node[above, midway, rotate=90]{ID Proof}; 
  \draw[line, line width=0.05mm] (e1.south)-- (Can1.north) node[above, midway, rotate=270]{Unique ID};

  \node[circle, line width=0.1mm, draw=black, fill=red, text=black] (Can2) at (-24,-3.5) {C};
  
  \draw[line, line width=0.1mm] (Can2.north)-- (e2.south) node[above, midway, rotate=90]{ID Proof}; 

 \node[draw=none] (t1) at (-27,-4.25){};
 
\node[draw=none] (t2) at (-23.05,-4.25){};
 \draw[line, line width=0.05mm] (t1.north) -- (t2.north)node[above, midway]{$t_{bcr}-t_{ecr}$};
 \draw[line, line width=0.05mm] (t2.north) -- (t1.north)node[above, midway]{};

  \node[circle, line width=0.1mm, draw=black, fill=yellow, text=black] (Can3) at (-22,-3.5) {V};
  \draw[line, line width=0.05mm] (Can3.north)-- (e3.south) node[above, midway, rotate=90]{ID Proof}; 
  \draw[line, line width=0.05mm] (e3.south)-- (Can3.north) node[above, midway, rotate=270]{Unique Token};
\node[draw=none] (t3) at (-23,-4.25){};
 \node[draw=none] (t4) at (-21,-4.25){};

  \draw[line, line width=0.05mm] (t3.north) -- (t4.north)node[above, midway]{$t_{btd}-t_{etd}$};
 \draw[line, line width=0.05mm] (t4.north) -- (t3.north)node[above, midway]{};

     \node[circle, line width=0.1mm, draw=black, fill=brown, text=black] (Can4) at (-18,-3.5]) {$\mathcal{A}$};
  \draw[line, line width=0.1mm] (Can4.north)-- (f1.south) node[above, midway, rotate=90]{Encryption Key}; 
  
    \node[circle, line width=0.1mm, draw=black, fill=brown, text=black] (Can5) at (-16,-3.5) {$\mathcal{A}$};
  \draw[line, line width=0.1mm] (Can5.north)-- (f2.south) node[above, midway, rotate=90]{Get Candidate List}; 
  
    \node[circle, line width=0.1mm, draw=black, fill=yellow, text=black] (Can6) at (-14, -3.5) {V};
  \draw[line, line width=0.1mm] (Can6.north)-- (f3.south) node[above, midway, rotate=90]{Token \& Encrypted Vote}; 
  
 \node[draw=none] (t5) at (-19,-4.25){};
 \node[draw=none] (t6) at (-13.05,-4.25){};
 \draw[line, line width=0.05mm] (t5.north) -- (t6.north)node[above, midway]{$t_{bvc}-t_{evc}$};
 \draw[line, line width=0.05mm] (t6.north) -- (t5.north)node[above, midway]{};

   \node[circle, line width=0.1mm, draw=black, fill=brown, text=black] (Can7) at (-12,-3.5) {$\mathcal{A}$};
  \draw[line, line width=0.1mm] (Can7.north)-- (f4.south) node[above, midway, rotate=90]{Encryption Key}; 
  
    \node[circle, line width=0.1mm, draw=black, fill=brown, text=black] (Can8) at (-10,-3.5) {$\mathcal{A}$};
  \draw[line, line width=0.1mm] (Can8.north)-- (f5.south) node[above, midway, rotate=90]{List of Votes}; 
  
 \node[draw=none] (t7) at (-13,-4.25){};
 \node[draw=none] (t8) at (-9,-4.25){};
 \draw[line, line width=0.05mm] (t7.north) -- (t8.north)node[above, midway]{$t_{bvt}$};
 \draw[line, line width=0.05mm] (t8.north) -- (t7.north)node[above, midway]{};
  
 \node[draw=none] (T1) at (-26,-5){};
 \node[draw=none] (T2) at (-10,-5){};
 \draw[line, line width=0.5mm] (T1.north) -- (T2.north)node[above, midway]{\textbf{\proto\ Protocol Timeline}};

    \node[draw=none] (x) at (-23.75,1){OFF-CHAIN METHODS};
    
    \node[draw=none] (x) at (-14,1){ON-CHAIN METHODS};

\end{tikzpicture}
\caption{\label{fig:sec_comaprison} Illustration of the protocol timeline in \proto. Here, C: Candidate, V: Voter, and $\mathcal{A}$: Other agent}\label{fig:voting}
\end{figure*}
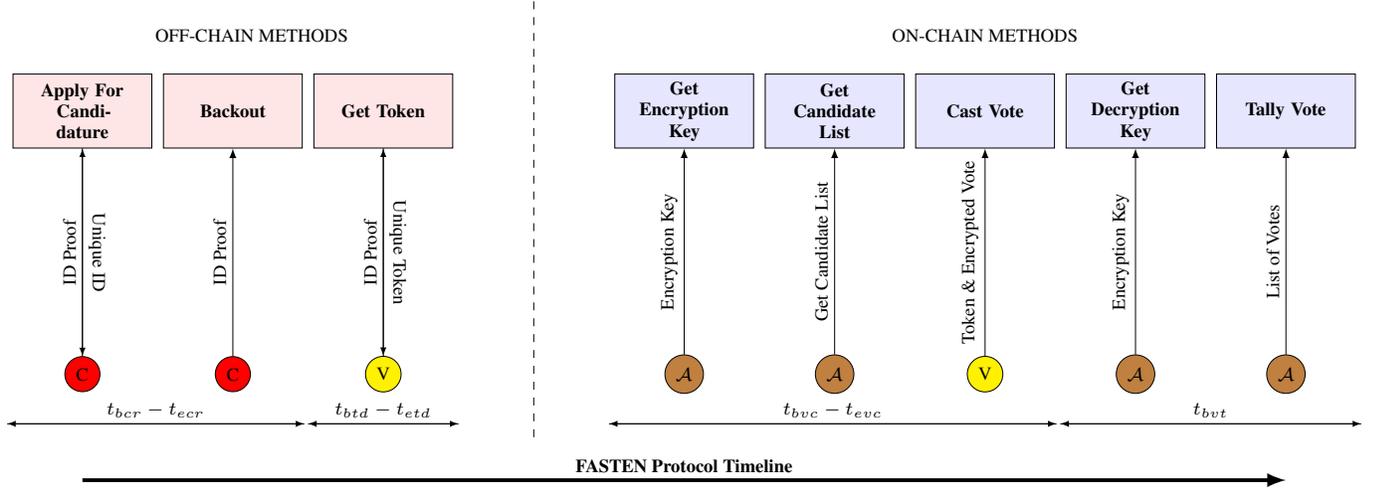
\normalsize
        %

We use an Ethereum-based smart contract for on-chain methods which in turn relies on the blockchain technology. Hence, we first summarize blockchain technology and Ethereum-based smart contracts in the following subsections.  

\subsection{Blockchain}
A blockchain can serves as an open, append-only distributed ledger that can record transactions between two parties efficiently, verifiably and in a permanent way. 
Blockchain security method comprise public-key cryptography. A public key (a long, random-looking string of numbers) is an address on the blockchain. Value
tokens sent across the network are recorded as belonging to that address. 
Data stored on the blockchain is generally considered incorruptible. Therefore, blockchains are secure by design. The interested readers are referred to \cite{narayanan2016bitcoin} for more details.

To maintain this ledger, blockchain comes with stack based programming that triggers transactions; and uses ad-hoc message passing and distributed networking. Its network does not have centralized points of vulnerability that computer hackers can exploit nor any central point of failure. 
This makes blockchain potentially suitable for the recording of events, medical records, and other management activities, such as identity management, transaction processing, and documenting provenance \cite{wiki:blockchain}. By storing data across its network, the blockchain eliminates the risks that come with data being held centrally. This leads to \emph{smart contracts}. In particular, we use Ethereum-based smart contracts.

\subsection{Ethereum-based Smart Contracts}

Ethereum is an open blockchain platform that allows anyone to build and use \emph{decentralized applications} (DApps) that run on blockchain technology.
Ethereum is not limited to predefined operations, as in the case of Bitcoin, but allows execution of user-defined operations written in computer codes/scripts on the blockchain through \emph{Ethereum Virtual Machine} (EVM). Thus, Ethereum supports the execution of \textit{smart contracts}.


The concept of \emph{smart contracts} was introduced in 1994 by Nick Szabo \cite{szabo96sc}. A smart contract is similar to a legal contract in adhering to strict rules and consequences, stating the obligations, benefits and penalties which may be due to either party in various circumstances. In addition to all this, it can take information as an input, process that information using the set of rules defined in the contract, and take action accordingly.

A smart contract in Ethereum is executed on EVM. For every executed instruction, there is a specified cost, expressed in the number of \textit{gas} units. Gas is the name for the execution fee that senders of transactions need to pay for every operation made on an Ethereum blockchain. Gas and ether are decoupled deliberately since units of gas align with computation units having a natural cost, while the price of ether generally fluctuates as a result of market forces. The Ethereum protocol charges a fee per computational step that is executed in a contract or transaction to prevent deliberate attacks and abuse on the Ethereum network \cite{ethereum_docs}.

We use the Ethereum smart contracts to enforce \proto. We now explain important aspect in designing of voting protocol to improve fairness in elections.

\subsection{\proto: Our Approach}

As a smart contract resides on blockchain, we inherent the decentralized feature of blockchain with it in our protocol. In designing a FSE, it is mandatory to verify the eligibility of a voter before allowing it to vote. This verification process requires examining the voter's real-world identity. Examining the same on-chain risks the privacy/anonymity of the voter. This is because it links the real-world voter identity with the encrypted vote. Therefore, there is a need for separating voter ID from encrypted vote. For the same purpose, we take help from the respective EC to distribute tokens which separates voter's real ID from vote. The tokens are randomly distributed off-chain (off the blockchain, that is not part of the contract) to voters after verifying their identity and eligibility. Since time is of utmost importance in an election, the voting protocol strictly follows a timeline of events to avoid any liabilities. Every method of the contract has a time bound out of which the contract \textit{reverts} the call.

\smallskip
\noindent\textit{Role of Election Commission (EC).}
To ensure that only eligible candidate gets to vote, we take help from EC to verify voter eligibility. Every data on the smart contract is on blockchain and thus is publicly available. This implies that the voter database cannot be made public to preserve voter anonymity. Also verifying voter ID on blockchain and registering vote through that ID violates voter anonymity. To resolve these, we let EC verify voter ID off-chain and distribute random tokens to the voters as a grant to vote. Note that every voter must be given only one token. The smart contract assumes that EC distributes the token truthfully, securely and privately without keeping any trace of voter ID with the token. The smart contract holds a pre-stored database of hashes of the token issued by the EC. Since the token is distributed by verifying the voter eligibility, the smart contract need only check if the token is genuine and not fake. It does so by checking if the hash of the token exist in the token hash database.


Also, the voters in \proto\ are not bothered with any commitments through deposits as in previously proposed voting mechanisms based on blockchain. In contrast, they are simply required to get the encryption key and ID, encrypt the vote and send it through the smart contract along with their token. The overhead of submitting the decryption keys relies with a different set of agents. We refer to such a third party as a \textit{Warden}. These agents store the decryption key safely instead of the voters, the process of which is explained later. Each warden holds a key which corresponds to a batch of vote. These wardens represent a distributed trust system and are incentivized to act honestly in the system. 

We also assume the candidate registration to be handled by EC and the candidate list to be provided by them in advance.

\smallskip
\noindent\textit{Role of Smart Contracts.} EVM(s) which are currently used for conducting election systems are black box to the voter because the software used inside EVM cannot be inspected by them. This raises concerns regarding code tempering by EVM. Unlike the code in EVM, a smart contract's code resides on a public distributed ledger, i.e., blockchain. This enables anyone to inspect the method's deployed by a smart contract. Also since smart contract resides on blockchain, every data/transaction that is part of the contract also becomes a part of the blockchain, which ensures the permanence of the history of the vote done by the voter. Thus, it results in the \textit{verifiability} of the vote counting.

After time $t_{evc}$, when the decryption key is released public, it can be used to decrypt the encrypted votes stored safely on blockchain and compare with the result achieved by the smart contract. Since the encryption keys and corresponding decryption keys have to be different by the design of the protocol, we use an asymmetric cryptography. Any asymmetric cryptographic system will work for our protocol such as 
ElGamal cryptography system~\cite{elgamal1985public}. With these as a backdrop, we now formally present our FSE protocol, \proto.

\section{\proto: Protocol Design}
In \proto, we use Ethereum-based smart contracts to enforce the protocol. The contract maintains strict adherence to the time constraints. Therefore each method in the contract is bound by a time window outside of which the contract reverts the call made to it. These time-constraints are defined and explained in Table~\ref{tab:TC}.

%
\begin{table}[!t]
    \centering\renewcommand*{\arraystretch}{1.25}
    \begin{tabular}{cc}
    \textbf{Notation}     &  \textbf{Definition} \\
    \hline
    \multirow{2}{*}{$t_{bcr}$}     &  time-stamp for beginning the candidature registration \\
    &  (Candidates can begin registration after this time) \\
    \multirow{2}{*}{$t_{ecr}$}     &  time-stamp for ending the candidature registration \\
    &  (All candidates must register before this time) \\
        \multirow{2}{*}{$t_{btd}$}     &  time-stamp for beginning the token distribution \\
    &  (Voters can collect their tokens after this time) \\
            \multirow{2}{*}{$t_{etd}$}     & time-stamp for ending the token distribution \\
    &  (All voters must collect their token before this time) \\
                \multirow{2}{*}{$t_{bvc}$}     & time-stamp for beginning the vote casting \\
    &  (Voters can begin casting their vote after this time) \\
                    \multirow{2}{*}{$t_{evc}$}     & time-stamp for ending the vote casting \\
    &  (All voters must cast their vote before this time) \\
                        \multirow{2}{*}{$t_{bvt}$}     & time-stamp for beginning the vote tally \\
    &  (Anyone can ask for vote tally only after this time) \\
    \end{tabular}
    \caption{Time Constraints in \proto}
    \label{tab:TC}
\end{table}
%

%
\begin{table}[!t]
    \centering\renewcommand*{\arraystretch}{1.25}
    \begin{tabular}{cc}
    \textbf{Variables}     &  \textbf{Definition} \\
    \hline
     \textit{block.timestamp} & Current time-stamp \\
    \textit{candList} & List of the registered candidates \\
     \textit{numKeys} & Total number of the encryption keys available\\
     \textit{idCounter} & To allocate the encryption-decryption key pair \\
     \textit{enKeys} & List of all encryption keys provided \\
      \textit{deKeys} & List of decryption keys aggregated by the wardens \\
      \textit{hashDatabase} & Database of the hashes of all the tokens \\
      \textit{voteBatch} & Database to store votes \\
      \textit{securityAmt} & Security amount wardens deposit before registration \\
      \textit{refundAmt} & Array to store amount to be refunded to the wardens \\
    \textit{reward} & Additional reward to the wardens \\ 
    \textit{wardens} & Mapping from warden address to key id \\
    \textit{sampleText} & Sample file to authenticate key pair(s) \\
     \textit{tallyDone} & Flag to ensure votes are only counted once\\
    \end{tabular}
    \caption{Variables in \proto}
    \label{tab:Var}
\end{table}
%

\proto\ also uses some predefined variables and databases, which are fed into the smart contract beforehand. We describe them as follows. 

\smallskip
\noindent\textit{\proto: Variables.} Table~\ref{tab:Var} provides the variables in \proto\footnote{For older versions of Solidity, i.e., $< 0.7.0$, \textit{block.timestamp} corresponds to \textit{msg.now}.}. From Table~\ref{tab:Var}, \textit{hashDatabase} is basically a mapping from hash string to boolean such that $True$ value means the hash is present in database. It is pre-populated beforehand. Further, \textit{voteBatch} is a double array where $voteBatch[i]$ represents the array of encryption votes with encryption id $i$. The list of encrypted votes gets aggregated in the course of the election. 
These variables are part of the underlying methods deployed in \proto. We now describe these methods next.



\subsection{\proto: Underlying Methods}
In \proto, for improved scalability and efficiency, we consider a \textit{hybrid} model, i.e., some of the underlying methods are on-chain while others are off-chain. Later, we describe the control flow of the protocol using these methods.

\subsubsection{Off-Chain Methods} These are the methods that are not part of the smart contract. We formally present them next.
\begin{itemize}
\item \textbf{ApplyForCandidature.} This method allows eligible candidates to register themselves as the candidate for election. For the implementation, we suggest the use of a biometric scanner for ID authentication. The method generates a unique voting ID for the candidate. It can only be used during $(t_{bcr}, t_{ecr})$.

\item \textbf{Backout.} Registered candidates can use this to \textit{backout} from the election. The candidates must first authenticate their identity using the biometric scanner. The method can also only be used during $(t_{bcr}, t_{ecr})$.

\item \textbf{GetToken.} This method is not part of the smart contract but a facility provided outside the contract to get the token privately. This method will distribute a unique token which the voter can use to cast vote. The token will only be distributed after a succesful ID authentication through biometric. It can be used only in the duration $(t_{btd}, t_{etd})$.
\end{itemize}

\subsubsection{On-Chain Methods}: These set of methods comprise those that are part of the smart contract. We classify them under two subsets: ``General public" methods and ``Warder Specific" methods. This categorization is based on whom all can invoke these methods. We now formally present them.

\smallskip
\noindent\textit{General public methods.} These include,

\begin{itemize}
\item \textbf{GetCandidateList.} This method returns the list of all the candidates participating in the election. This is simply a \textit{getter} method. For e.g., a function that is set to ``view" in solidity. The method requirese no authentication and can be invoked during $(t_{ecr},t_{bvc})$.


\item \textbf{GetEncryptionKey:} This method shares the encryption key with the caller who invokes it. It also requires no authentication. The caller is provided with an encryption ID and the corresponding encryption key. It can be invoked only in the duration $(t_{bvc}, t_{evc})$. We present the pseudo-code in Method~\ref{method:GEK}.

\begin{algorithm}[!t]
\DontPrintSemicolon
\renewcommand{\algorithmcfname}{Method}
require($block.timestamp \in (t_{bvc}, t_{evc})$) \;
Int $i \leftarrow idCounter+1$ \;
$idCounter = (idCounter + 1) \% numKeys$ \;
bytes32 $ek \leftarrow enKeys[i]$ \;
return $[i, ek]$ \;
\caption{\label{method:GEK}GetEncryptionKey}
\end{algorithm}

\item \textbf{CastVote:} Presented in Method~\ref{method:CV}, it allows the eligible voters to register their vote. Required parameters are token\footnote{This is the only time the voter is required to send a token to the smart contract.}, encryption id and encrypted vote. The token is validated after which the encrypted vote is stored corresponding to its encrypted id. This method can be used only in the duration $(t_{bvc}, t_{evc})$.

\begin{algorithm}[!t]
\DontPrintSemicolon
\renewcommand{\algorithmcfname}{Method}
\KwIn{Token $t$, Int $i$, bytes32 $ev$}
require($block.timestamp \in (t_{bvc}, t_{evc})$) \;
bytes32 $h \leftarrow  sha3(t)$ \;
require($hashDatabase[h] == true$)
$hashDatabase[h] = false$ \;
$voteBatch[i].push(ev)$ \;
\caption{\label{method:CV}CastVote}
\end{algorithm}

\item \textbf{GetDecryptionKeys.} Similar to Method~\ref{method:GEK}, this method shares the list of decryption keys with the caller who invokes it. The method requires no authentication \textit{but} can only be invoked after $t_{bvt}$. 

\item \textbf{TallyVote:} This decrypts the votes, counts them and returns the result (Method~\ref{method:VT}). It requires no authentication. Once the votes have been decrypted and counted, it returns the pre-computed result in the subsequent calls. This method can also only be invoked $t_{bvt}$.

\begin{algorithm}[!t]
\DontPrintSemicolon
\renewcommand{\algorithmcfname}{Method}
require($block.timestamp > t_{bvt}$) \;
\If{$tallyDone == false$} {
\For{Int $i = 0; i < numKeys; i++$} {
\For{bytes32 $ev: voteBatch[i]$} {
bytes32 $dk = deKeys[i]$
bytes32 $dv \leftarrow decrypt(ev)$  \;
bytes32 $candId = extract(dv)$ \;
$candTally[candId] += 1$ \;
}
}
$tallyDone = true$ \;
}
return $candTally$
\caption{VoteTally\label{method:VT}}
\end{algorithm}

\end{itemize}

\noindent\textit{Warden Specific Methods.} These include,

\begin{itemize}
\item \textbf{DepositSecurity.} The wardens invoke this method to deposit monetary value as security against their honest behavior and to confirm their registration. The method is presented in Method~\ref{method:DS}.
\begin{algorithm}[!t]
\DontPrintSemicolon
\renewcommand{\algorithmcfname}{Method}
require($wardens[msg.sender] > 0$) \;
require($block.timestamp < t_{bvc}$) \;
require($msg.value > securityAmt$) \;
$refundAmt[msg.sender] = msg.value - securityAmt$ \;
\caption{DepositSecurity\label{method:DS}}
\end{algorithm}

\item \textbf{SubmitEncryptionKey.} As described in Method~\ref{method:SEK}, this method allows the wardens to submit the encryption key in the required duration. 
\begin{algorithm}[!t]
\DontPrintSemicolon
\renewcommand{\algorithmcfname}{Method}
\KwIn{bytes32 $ek$}
uint $id = wardens[msg.sender]$ \;
require($id > 0$) \;
require($block.timestamp < t_{bvc}$) \;
require($refundAmt[msg.sender] > 0$) \;
$enKeys[id] = ek$ \;
\caption{SubmitEncryptionKey\label{method:SEK}}
\end{algorithm}

\item \textbf{SubmitDecryptionKey.} This method is used by the wardens to timely submit the correct decryption key. We present the method in Method~\ref{method:SDK}.
\begin{algorithm}[!t]
\DontPrintSemicolon
\renewcommand{\algorithmcfname}{Method}
\KwIn{bytes32 $dk$}
uint $id = wardens[msg.sender]$ \;
require($wardens[msg.sender] > 0$) \;
require($block.timestamp \in (t_{evc}, t_{bvt})$) \;
require($refundAmt[id] > 0$) \;
bytes32 $ek = enKeys[id]$ \; 
require($sampleText == decrypt(encrypt(sampleText, ek), dk)$) \;
$deKeys[id] = dk$ \;
refundAmt[msg.sender] += securityAmt + reward \;
\caption{SubmitDecryptionKey\label{method:SDK}}
\end{algorithm}

\item \textbf{WithdrawReward:} From Method~\ref{method:WR}, wardens can use this method to collect their reward after successful submission of decryption key.
\begin{algorithm}[!t]
\DontPrintSemicolon
\renewcommand{\algorithmcfname}{Method}
require($wardens[msg.sender] > 0$) \;
require($block.timestamp > t_{bvt}$) \;
uint $amt = refundAmt[msg.sender]$ \;
$refundAmt[msg.sender] = 0$ \;
\If{amt \textgreater  0}{$msg.sender.transfer(amt)$}
\caption{WithdrawReward\label{method:WR}}
\end{algorithm}

\end{itemize}

As aforementioned, in \proto, we use tokens distributed by EC to register a voter's vote. We now explain how we use those tokens to guarantee secure voting.

\subsection{\proto: Token Validation Process}

The tokens are distributed through an off-chain system. Token get pre-generated and stored securely and privately with EC. We assume that tokens will be distributed privately and securely by the respective EC. The voter provides the token only once while sending its encrypted vote. The smart contract keeps the hash\footnote{We suggest hash functions which are efficient over a smart contract such as \textit{sha3} \cite{wood14ethereum}.} of all the tokens in its on-chain database. It calculates the hash of the token provided by the voter and checks if the hash of the token exist in the database. If found, it removes the entry of the token from the database and registers the voter's vote.

\begin{algorithm}[!tb]{}
\DontPrintSemicolon
\renewcommand{\algorithmcfname}{Procedure}
\underline{\textbf{Vote Casting Window Opens}}\;
token $t \leftarrow getToken()$ \;
Candidate List $candList \leftarrow GetCandidateList()$ \;
Let $canList$ be the preferred candidate of Gretel from the list $candList$ \;
Encryption ID, Encryption Key $i, ek_{i} \leftarrow GetEncryptionKey()$ \;
Encrypted Vote $v \leftarrow Decrypt()$, encryption done off-chain by the voter \;
$CastVote(t, v, i)$\;
\underline{\textbf{Vote Casting Window Ends}} \;
\underline{\textbf{Vote Tallying window opens}}\;
Vote Count[] $vc \leftarrow TallyVote()$ \;
$vc_{i}$ represents vote count of $i^{th}$ candidate in the $canList$
Decryption Key List $dk \leftarrow GetDecryptionKeys()$ \;
where $dk_{i}$ represents decryption key associated with the encryption id $i$ \;
The list of decryption key can be used to get the vote count manually check the same with calculated by the contract.
\caption{Voting Procedure()}\label{voting-procedure}
\end{algorithm}

\subsection{\proto: Voting Procedure}
Now that we have provided the core protocol design, we give a sample walk-through of the protocol. A voter who wishes to use \proto, can use it by following Procedure~\ref{voting-procedure}.

We will now describe the way through which the decryption key is kept secret in \proto.

\subsection{\proto: Warden Assistance}
As a smart contract cannot store data privately, we take we take assistance from wardens outside the contract to timely submit the decryption keys. In \proto, wardens are appointed to keep the decryption keys off the chain. They provide the keys to smart contract through a special transaction when the time comes, i.e., during $(t_{evc}, t_{bvt})$.

Consider a set of wardens $W$. Each warden $w_{i}\in W$ has a decryption key $dk_{i}$ and is assigned a batch of votes which can be decrypted through that key. Trivially, $|W|$ are the total number of batches. Now, let $B$ be the dataset of batches, i.e., where $b_{i}\in B$ is the batch corresponding to $w_{i}$. Observe that every time GetEncryptionKey() (Method~\ref{method:GEK}) is called, the contract provides the voter with the encryption id $i$ corresponding to $w_{i}$ and an encryption key $ek_{i}$. The voter encrypts the vote using the encryption key and then sends the encrypted vote and ``$i$" as parameters to  CastVote() (Method~\ref{method:CV}). The contract then assigns the encrypted vote to $b_{i}$. As a result, after $t_{evc}$, if we have $n$ votes in total, every batch $b_{i}$ will hold roughly $\frac{n}{|W|}$ number of encrypted votes.

In \proto, we do not take for granted that the wardens will be honest. Towards this, we design our protocol so as to minimize the loss in case any warden turns out to be dishonest. Observe that any dishonest warden can cause trouble in two ways: aborting its duty and/or leaking the decryption key. We now provide our solution to these two problems.

\subsubsection{Abortion of Duty} Each warden $w_{i}$ submits a deposit $dk_{i}$ prior to its participation in \proto. The warden's address will be stored in the smart contract database. As a result, it will get its deposit refunded only if it timely submits the correct decryption key. The smart contract will verify the decryption key by trying the encryption-decryption key pair on some $r$ random strings. Once the decryption key is verified to be correct, the warden will get the deposit refunded. We also suggest to provide an additional bonus amount as a reward for honest behavior.

\subsubsection{Leaking Decryption Key} We remark that it is practically impossible to prevent any warden $w_{i}$ from leaking its decryption key $dk_{i}$. In the event that a warden $w_{i}$ leaks the decryption key, on average, it will leak the tally of only $\frac{n}{|W|}$ votes. Thus, we suggest to choose the number of wardens $|W|$ so as to minimize this loss to a very minute percentage. For example, having less than $\frac{1}{1000}$\textsuperscript{th} fraction of voters with each warden, will ensure that a single dishonest warden cannot leak tally of more than $0.1\%$ of votes. The number of wardens is a protocol parameter and can be increased depending on the budget and risk tolerance.

\section{\proto: Proof of Fairness and Security}
As aforementioned, an election is said to be fair if it satisfies the four properties: vote anonymity, vote concealment, vote immutability and double voting inhibition. That is, every voter is able to vote discretely and anonymously and the vote tally remains hidden until the end of the election. Further, nobody is able to cast the vote without authentication nor cast the vote twice. We now (i) provide the formal definitions; and (ii) prove that \proto\ satisfies all these properties.

\begin{definition}[Voter Anonymity] \label{claim:5.1}
An election protocol is said to satisfy the \textit{Voter Anonymity} property if probability of finding the vote of any voter is negligibly small as compared to the size of the population.
\end{definition}
\begin{claim}
\proto\ satisfies voter anonymity with probability of guessing any voters token being at most $\frac{n}{2^l}$ where $n$ is the number of voters and tokens are $l$-bit long.
\end{claim}
\begin{proof}
The smart contract takes the token as an attribute from the voter to register their encrypted votes. We demand the voters to use a new Ethereum address to cast their encrypted vote since the old Ethereum address might have been compromised by the voter. Since the token is randomly generated and privately distributed off the blockchain, it cannot be linked back to the voter through the blockchain because Ethereum addresses are randomly generated and have no relation with the identity of the user. The public ledger will only show that a certain encrypted vote has been registered for a certain token. It implies that even after the vote is encrypted, it is related to only the token and has no links with the previous holder of the token. Thus the vote can never be traced back to the voter except possibly random guessing. If token is $l$-bit long, and size of the voter population is $n$, the random guessing will not succeed by probability more than $\frac{n}{2^l}$. By choosing, $l >> \log_2 n $, we can ensure voter anonymity.
\end{proof}

\begin{definition}[Vote Concealment] \label{claim:5.2}
An election protocol is said to satisfy the \textit{Vote Concealment Property} if no vote's value is revealed before the end of the election vote casting period with probability more than $\frac{k}{1000}$ if no more than $k$ wardens are dishonest.
\end{definition}

\begin{claim}
\proto\ satisfies vote concealment property. 
\end{claim}
\begin{proof}

As we have highlighted earlier, the voters will register their encrypted vote to the contract. The encryption key and the corresponding encryption ID will be provided by the contract itself to the voters. The voter encrypts the vote through the encryption key off the blockchain on her device and sends it to the smart contract along with the encryption ID and the respective token. Since the vote is encrypted off the chain, the real value of the vote remains hidden until the decryption key is out. The decryption is made public only after the end of vote casting window, thereby concealing the vote until the voting period ends. Hence the property of Vote Concealment is preserved.
\end{proof}

\begin{definition}[Vote Immutability] \label{claim:5.3}
An election protocol is said to satisfy the \textit{Vote Immutability} property if no vote's value can be altered once it has been casted.
\end{definition}
\begin{claim}
\proto\ satisfies vote immutability under assumption that majority of the nodes in the network are honest.
\end{claim}
\begin{proof}
Ethereum blockchain stores all the transactions permanently by default. Once a transaction is part of the public ledger, it cannot be removed or changed unless 51\% or more nodes are corrupt. Also, since once the hash of a token is matched to that of the database on the blockchain, it is immediately removed, thereby no overwriting of the vote is possible.
\end{proof}

\begin{definition}[Double Voting Inhibition] \label{claim:5.4}
An election protocol is said to satisfy the \textit{Double Voting Inhibition} property if probability that any voter can cast more than one vote in the election is negligible.
\end{definition}
\begin{claim}
\proto\ satisfies double voting inhibition property.
\end{claim}
\begin{proof}
The smart contract has an on-chain database containing the pre-calculated hashes of the tokens distributed to the voters. The voters are required to send their tokens to register their vote. The smart contract calculates the hash of the received token and compares it with the ones stored in the database. If a match is found, the token hash is removed from the database and the vote is registered. Since the token is immediately removed from the database before registering the encrypted vote, we can safely claim that every token is used only once to register the vote, thus inhibiting double voting. The only possibility for an voter to double vote would be to cast a vote with its valid token as well as guess a random token whose hash matches with that one in database. If hash is of length $h$-bits, even by birthday-paradox the voter needs to try $2^{\frac{h}{2}}$ trials to guess a first valid token apart from its own. Typically hash sizes ($h$) are 256 bits. Hence, the probability of successful double voting is negligible.
\end{proof}

Thus, we have proved theoretically that \proto\ achieves fairness in election, i.e., \proto\ is a solution for FSE. Now we perform cost analysis of \proto.

\section{\proto: Protocol Analysis}
\subsection{Cost Analysis}
We analyze the cost in terms of ethereum gas, which is constant cost of network resources/utilisation. We use the gas estimation given in ethereum docs and ethereum rate card also given in the docs for the estimation of cost per vote \cite{}. We estimate the cost per vote in 2 stages: (i) Voter side cost; and (ii) Warder side cost.

\subsubsection{Voter Side Cost} In \proto, each voter has to follow a particular sequence of methods in order to cast the vote. From Procedure~\ref{voting-procedure}, the voter follows the given sequence: (i) GetCandidateList, which has its gas requirement as 26 units; (ii) GetEncryptionKey, which consumes 667 gas units; and (iii) CastVote, that consumes 739. In total, the vote casting consumes 1432 Gas units.

Further, TallyVote computes the vote count for each candidate when called for the first time and returns the result. After this, it returns the pre-computed result for subsequent calls. This optimization ensures that there are no redundant calculations. Therefore, every vote is decrypted and counted once. This ensures that TallyVote is equivalent to counting of a single vote, which we estimate as $300000$ gas units. The estimation is high as we encrypt votes using $160$-bit ElGamal encryption. Thus, the decryption cost is itself high. 

Fortunately, there is an easy optimization to this as well. We suggest to decrypt the vote on DApp instead of the smart contract. As smart contracts are mainly used for maintaining persistent states throughout the network and the decryption won't alter the state, the decryption can be done on the DApp. By avoiding the decryption cost, we get an upper bound of $1500$ gas units cost per vote for the voter side.

\subsubsection{Warden Side Cost} In \proto, each warden uses the following methods in the order given: (i) DepositSecurity, which requires 23 gas units; (ii) SubmitEncryptionKey, which requires 629 gas units; (iii) SubmitDecryptionKey, with 600755 gas units; and (iv) WithdrawReward, which requires 21629.  That is, each warden consumes $623036$ gas units in total. The reason for such high cost is the encryption and decryption operations done in \emph{SubmitDecryptionKey} method to check the authenticity of the decryption key. But as mentioned during the analysis for voter side's cost, one can avoid the cost of encryption and decryption by computing it on the DApp instead of a smart contract. This also reduces the warden side cost to an upper bound of $23036$ gas units. Suppose there exists total of $n$ voters, then every warden holds the decryption key for $\frac{n}{|W|}$ voters, on average. Therefore, cost per vote from warden's side comes out to be $\frac{23036\cdot |W|}{n}$ gas units.

For a reasonable choice of $\frac{n}{|W|}$, such as 1000, this cost comes to be $\approx23$ gas units. Adding this to the voter side cost calculated above, we can set the upper bound for the total cost per vote of in \proto\ in terms of Ethereum gas units as $1600$ (when $\frac{n}{|W|}=1000$). Further note that, $1600$ gas units corresponds to $0.000064$ ETH as 1 gas unit equals $4\cdot 10^{8}$~\cite{gas-eth}. As at the time of writing of this paper, we have $1\mbox{~ETH~}\approx 627$ USD, the overall cost comes out to be $0.040$ USD per vote.

\subsection{Time Analysis}
As shown, the total gas consumption in \proto\ is $1600$ per vote (when $\frac{n}{|W|}=1000$). The block gas limit at the time of writing of paper, for Ethereum, is $8*10^{6}$~\cite{block-gas-limit}. Thus, each Ethereum block can hold $5000$ votes. Further, at the time of writing this paper, each block takes $15s$ to get on the Ethereum Network~\cite{block-time}. This corresponds to a processing capacity of roughly $20000$ votes per hour\footnote{We remark that this estimate can be significantly improved by deploying a more scalable underlying blockchain consensus protocol such as \cite{eos_whitepaper,dfninity,Arora2020}.}.

\smallskip
\noindent\textit{Inference from Analysis.}
A rough estimate of the cost per vote is $0.04$ USD (when $\frac{n}{|W|}=1000$). This is significantly lower than the amount of money spent in countries all over the world for mass elections. For Example, the per vote election cost for UK European Parliament Election 2014 is $5.54$ USD \cite{uk-elections}.

\section{Conclusion and Discussion}
We showed that \proto\ is a solution for FSE through blockchain and smart contracts. It also provides a transparent decentralized system through which the stakeholders can verify the result. We also argued that our protocol is cost efficient compared to the existing election methods. In summary,

\begin{enumerate}[leftmargin=*]
\item \proto\ does not reveal voter identity at any stage of the process.
\item In \proto, votes cannot be revealed at any time before the vote casting window expires. Even the vote count is concealed from everyone until the window expires. Thus, \proto\ ensures that voting is not influenced at any stage.
\item \proto\ allows anyone to check the end-to-end voting procedure independently. This is because all the relevant information is on a public ledger. As proved, this is achieved without compromising voters' privacy as every vote is cryptographically secure and cannot be traced back to the voter.
\item As we do not assume nor use any constraints on voting population, \proto\ can be used for a large population. 
\item Our cost analysis shows that \proto\ is cost efficient as compared to exisiting protocols. 
\item Unlike prior works, in \proto\ voters do not have to commit and return to the protocol to reveal their votes.
\end{enumerate}

\subsection{Discussion}

\noindent\textit{Oracles: An alternative to Warden Assistance.}
As smart contracts cannot access and fetch data outside the blockchain, in \proto, we rely on wardens to provide the decryption keys. Another way to achieve this is through \textit{oracles}. An oracle is a third party service designed for smart contracts and can feed data from outside the blockchain to it, as and when required. However, relying on such third party oracles can compromise the distributed trust model of the underlying blockchain. As a result, in \proto, we leverage wardens. 

In order to use oracles, we need a way to ensure the data provided is genuine and not tempered. Oraclize \cite{oraclize} is one such service provider which claims to be \emph{provably honest}, i.e., which provides unaltered data to smart contracts. They do so by accompanying the returned data together with a document - referred as an ``authenticity proof" - which can be requested using the \textit{oraclize\_setProof} function provided by the service. The authenticity proofs build upon different technologies such as auditable virtual machines and trusted execution environments (refer \cite{oraclize-docs}).


\smallskip
\noindent\textit{Liquid Democracy.} The use of transferring voting rights to an informed voter has been referred as \emph{liquid democracy} \cite{kahng18}. Kahng \emph{et. al.} \cite{kahng18} showed existence of certain delegation mechanisms that can outperform traditional voting in terms of selecting better candidates.
We believe, such mechanisms can be easily implemented through \proto. It will need to do a transaction between the voter (willing to transfer voting rights) and the delegate to transfer its token. We leave security analysis of this for a  future work.

\bibliographystyle{unsrt}  
\bibliography{mybib}
\end{document}